\documentclass[10pt,twoside,twocolumn]{IEEEtran}
\usepackage{graphicx,cite,calc,xcolor,subfigmat,amssymb,amsmath,mathrsfs,dsfont,epstopdf,setspace,algorithm,algorithmic}
\usepackage[normalem]{ulem}

\newtheorem{theorem}{Theorem}

\title{Content Placement in Cache Networks Using Graph-Coloring}

\author{Mostafa~Javedankherad,~
        Zolfa~Zeinalpour-Yazdi,~\IEEEmembership{Member,~IEEE}

\thanks{Manuscript received xxxxx xx, xxxxx; revised xxx xx, xxxx.}
\thanks{The authors are with the Department of Electrical Engineering, Yazd University, Yazd, Iran  (e-mail:m.javedankherad@gmail.com, zeinalpour@yazd.ac.ir)}
}

\usepackage{algorithm}
\usepackage{algorithmic}
\begin{document}
\maketitle
\begin{abstract}
Small cell densification is one of the effective ideas addressing the demand for higher capacity  in cellular networks.  The major problem faced  in such networks is  the wireless backhaul link and its limited capacity. Caching most popular files in the memories of small cells base station (SBSs) is an effective solution to this problem.  One of the main challenges in caching is choosing the files that are going to be stored in the memory of SBS. In this paper, we model the described caching problem as a graph. This graph is divided into four sub-graphs including placement, access, SBSs and delivery graphs. By making some modifications to the SBS graph, we convert it to a graph that can be colored.  Coloring of the generated graph is NP-hard and we use an algorithm proposed in graph-coloring area to color it. To overcome the complexity of above coloring technique, we then propose a simple graph-coloring method based on two  point processes, Matern Core-type I and II. We model our network with a new weighted graph which simply can be colored and after that the files are cached accordingly.  We evaluate the performance of our proposed methods through simulations. Our results show that by employing out proposed method in a typical considered SBSs network, the load on the macro base station can be reduced by around $25 \%$, at the distribution parameter of popularity of files equals to $0.6$, compared to conventional policy which caches the most popular content every where.
\end{abstract}

\section{Introduction}
\IEEEPARstart{D}{ata} traffic on wireless networks has grown significantly in recent years as a result of the emergence of new communication devices such as tablets, mobile phones, etc., along with the increasing demand for watching online videos anytime and anywhere. The video files constitute the significant percentage of this data traffic. According to the Cisco's forecast, it can be expected that the demand for the traffic  of the whole increases to over 49 exabytes/month by 2021, and in particular  video files cover approximately 75 percent of the whole\cite{mj4}. These reasons have motivated researchers to look for ways to handle the burden imposed by video files on the network in the next generation of mobile networks. The network densification is one of the effective ideas posed to the traffic jump. Other ideas such as the spread of the radio spectrum in particular millimeter-wave communications, device-to-device communications (D2D) and massive MIMO techniques have also been proposed by the researchers to address this problem and they are currently  under investigation\cite{mj3, mj301,mj302,mj401}. Network densification deploys  a large number of heterogeneous base stations (BS) such as micros, picos and femtos to  improve spectral efficiency  and   the network coverage. This heterogeneous network is also called small cell networks (SCNs).
 Although better performance in terms of spatial reuse is obtained in dense SCNs, they face a major problem.   The problem is that the access points are usually connected to the network core with wireless backhaul links and therefore have limited capacity. Saving popular video files near user equipment(UE) can lead to the less  usage of backhaul links. In general, three major advantages are achieved through storing files in the SCNs\cite{mj1,kd2,mj3}:
  1) delay reduction in downloading popular files due to saving them in the  storages which are closer to the UE. 2) less need
  for expensive high capacity backhaul links 3) removing huge traffic from the main  macro base station (MBS) and distributing this load over SBSs.
Note that a significant amount of traffic is redundant in the network since UE are generally downloading popular duplicate videos from the BSs, whose redownloading from the main base stations can be avoided through this technique. In general, as the distance between the transmitter and the receiver gets shorter, the file gets received with less delay, lower cost and higher reliability. The technique of bringing popular files near UE and saving them in the SBSs is called caching.
 If the caching is properly designed, a significant improvement in reducing wireless network traffic, increasing download rates, and reducing file download latency results. Caching is done in two main phases: placement phase and delivery phase.
 The placement phase includes the embedding of popular video files at uncrowded network traffic hours, and the delivery phase includes transfer of the requested files to the UE\cite{mj6}. So one of the main issues is to perform the best file placement in the memory of SBSs in order to maximize the downloaded files from these memories, which is equivalent to the minimization of downloads from the MBS. Placement in the caches is done in both centralized and decentralized manner \cite{mj501}.

 The concept of graph coloring is one of the basics in graph theory, which can be used for our end. The coloring scheme of the graphs that began with conjecture of four-color by Arthur Cayley's  \cite{mj1400}, is one of the major branches in the graph theory. It is a graph indexing in which different tags are assigned to the adjacent vertices or adjacent edges.  This concept is used in different areas like scheduling, sudoku tables, assigning frequencies to the radio stations, register allocation and many other applications in the field of computer and other sciences.

In this paper, we intend to use graph coloring to cache popular files in the SBS' memories so that the average download of the files from the macro base station is minimized. Heterogeneous cellular networks are challenging when a UE is simultaneously covered by several SBS base stations. In this situation, the user has access to the files stored in the caches of all those SBSs. So if we do the placement phase so that SBSs which cover a user have fewer common stored files, it can be concluded that the user can access more files from the cache memory of those SBSs and less need for downloading them from the MBS. Our discussion becomes more beneficial when the degree of popularity of files and the rank of accessibility of caches are considered in loading the cache memory of SBSs. Since each user simultaneously has access to a larger number of files with a higher popularity, the probability of receiving the desired file from the SBS cache memory increases. In the caching scheme, we try to ensure that two SBSs which can serve the same users are not loaded by similar files and  coloring can do this for us.  This is the basis of  the placement in this paper. The problem of coloring, like the problem of caching, is an NP-hard optimization problem, and a number of algorithms for solving them are presented in \cite{mj20, mj14}.
\subsection{Related Work}
Caching has been a subject of interest to many researchers. Extensive research has been conducted on studying the
content-placement to obtain gains over the conventional policy which caches the most popular content everywhere \cite{kd1}. Different performance metrics have been considered to design the content-placement strategies
in the literature such as average downloading delay \cite{kd2, kd3, kd4}, average BER \cite{kd5}, average caching failure probability \cite{kd6}, offloading probability \cite{kd7}, density of successful receptions \cite{kd8}, hit probability \cite{kd6}, average success probability of content delivery \cite{kd9}, and the expected backhaul rate and the energy consumption \cite{kd10}.

Some papers address the placement problem for the more developed scenarios like hierarchical networks \cite{kd11}, multi-relay networks \cite{kd12} and  wireless cooperative systems \cite{kd13},  (including MIMO \cite{kd14} and  coordinated multi-point joint transmission (CoMP-JT) schemes \cite{kd15}). Also \cite{kd16, kd17} take into account the difference in users’ file preferences and develop a caching policy based on user preference profile. The authors in \cite{kd18, kd19} figure out effects of wireless fading channels  and show the trade off between the channel selection diversity  and file diversity and derive the content placement  considering the interactions among multiple users and network interference, and finally the mobility pattern of users is considered in \cite{kd20, kd21} to improve cache performance.
  Most of these scenarios are outcome of the problems which are proved to be NP-hard. Therefore, researchers have to sub-optimal techniques and practical algorithms  with  near optimal performance \cite{kd22} leveraged by the concepts such as greedy algorithms \cite{kd2} and the theory of belief propagation\cite{kd3}. Graph-coloring is another technique that can be used to assign files to different access points. The basic idea in the graph-coloring-based algorithm is to iteratively dye the uncolored vertices and color them with the color not used by neighboring vertices. While graph coloring ideas have been used  for resource allocation in D2D communication underlying cellular networks \cite{kd23, kd24, kd25}, to the best of our knowledge, this is the first work to use them to design algorithms for content placement in cache networks.

\subsection{Methodology ad Contributions}
The contributions of the paper are  stated shortly as follows:

We model the network as a graph. In order to do caching in this system, we transform the resulting graph into four graphs, including the placement, access, delivery and SBSs graphs. The SBSs graph is a general weighted graph so that we try to make new changes to this graph with respect to the placement, access and delivery graphs. We set a priority over the file stored in SBS's memories by doing the coloring process with the least possible colors. We do this using two different approaches.

Method 1:
We introduce a weighted graph in such a way that the weight of each edge is equal to the distance between two SBSs, corresponding to two vertices. To apply the coloring, we have to convert it into a graph without weight. To this end, we set the threshold in both universal and individual forms on the weight of edges and keep the edges with the weight below this threshold and remove the rest. We run a coloring algorithm on the resulting graph and prioritize coloring the vertices with higher degrees. We set the obtained colors for the vertices, which are numerical indices, as the priority of the filling of the corresponding SBS memories. We also sort the popular files  in a descending order. By filling the memories of the SBSs that correspond to the vertices with lower color index, we obtain  the placement graph.

 Method 2:
  We consider the position of SBSs as a Poisson point process. SBSs that share a common range with each particular SBS are placed as SBSs of the same class. Using two  point processes  Matern Core-type I and II, we assign weights to SBSs. Then we define the SBSs graph so that there is an edge between two vertices if the equivalent SBSs are  in the  same class and also assign to each vertex  the SBS weight corresponding to it. We apply coloring algorithm to this graph, which its coloring priority is to color  the vertex with a larger weight. We set the obtained colors for the vertices, which are numerical indices, as the priority of the filling of the correspond SBS memories. We sort the popular files in descending order. By filling the memories of the SBSs that correspond to the vertices have lower color index, we obtain the placement graph.

 The structure of the paper is as follows. System model and some necessary assumptions are given in Section II. Modeling the caching problem  as a graph and description of the coloring problem studied in Sections III. Sections IV and V present  two caching algorithms for coloring the graph. Simulation results are given in Section VI followed by a conclusion given in Section VII.
\section{System model and assumpthions}\label{shokr1}
\renewcommand{\labelenumii}{\arabic{enumi}}

Consider a heterogeneous  cellular network  with multi-tier structure comprising a MBS and SBSs such as picocells, femtocells and mobile users.

 In this model, it is assumed that SBSs are equipped with memories that can cache popular files. The macrocell covers the entire network, but a SBS covers a limited geographical range around itself. The sets $S=\{S_1,S_2,\cdots,S_{|S|}\}$ and $M=\{M_1,M_2,\cdots,M_{|S|}\}$ represent the SBSs and their memory sizes,  respectively. In this network, users who are introduced by $U=\{u_1,u_2,\cdots,u_{|U|}\}$ are mobile, unlike the BSs which are fixed. The structure of this network is shown in Fig. 1.

Assume that this network is clustered by the indices in  set $C=\{c_1,c_2,\cdots,c_{|C|}\}$ of size $|C|$. In each cluster, there are a number of SBSs and  mobile users.   The user can access several SCS at the same time and can be served by them. We put the popular files that are downloaded by users during a specific time period in a set called $F=\{f_1,f_2,\cdots,f_{|F|}\}$ of size $|F|$ which includes all popular files. Suppose the probability distribution of files is Zipf,  so the probability of requesting a file by a particular mobile user is given by
\begin{equation}\label{zipf}
P_f=f^{-\alpha}(\sum_{i=1}^{|F|}i^{-\alpha}),
\end{equation}
where $\alpha$ is the distribution parameter of Zipf function\cite{mj302}. We assume that the size of the files is equal which is not in contradiction with the reality as unequal files can be splitted into some slices of equal size.

\begin{figure}[t!]\label{shabake}
\begin{center}
\includegraphics[width=0.5\textwidth]{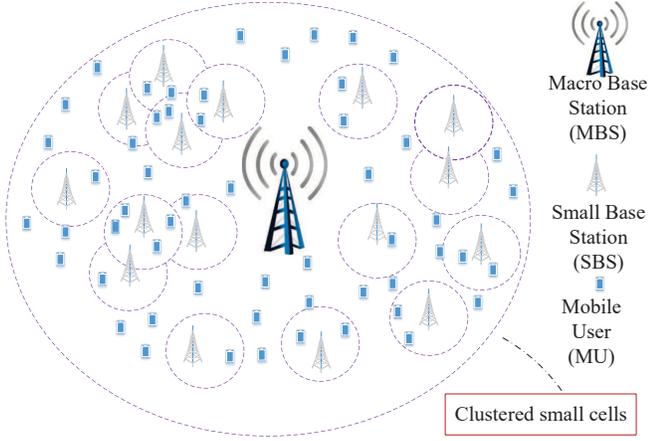}
\center{\caption{System model representation.}}
\end{center}
\end{figure}

\section{Graph-based model of the network}\label{mostafa1000}
In the following, we intend to model the considered network by  graphs so that we can apply the concepts implied in graph theory to them. We model the entire network as a graph like Fig. 2(a) where $\{S_1,S_2,\cdots,S_{|S|}\}, \{u_1,u_2,\cdots,u_{|U|}\}$ and $\{f_1,f_2,\cdots,f_{|F|}\}$ are SBSs, users, and files, respectively. It should be emphasized that this graph is a general graph and in a real network, some edges may not exist.
\begin{figure}[t!]\label{stratureNetwork2}
\begin{center}
\includegraphics[width=0.4\textwidth]{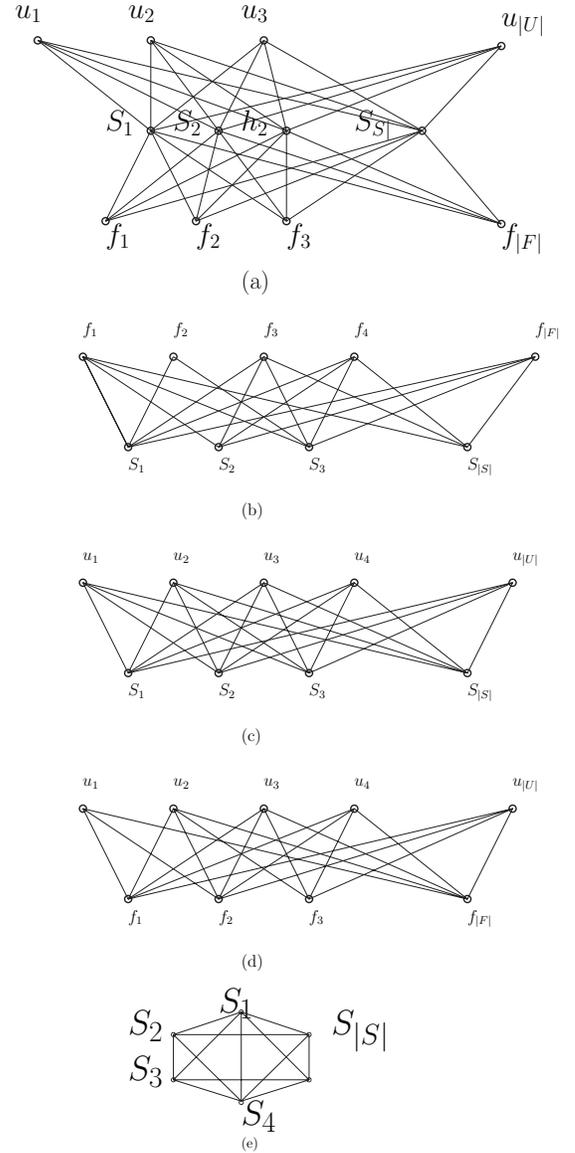}
\caption{(a)Whole network graph (b)Placement graph (c)Access graph (d)Delivery graph (e)SBSs graph.}
\end{center}
\end{figure}

To work more precisely on the graph, we cut it into four graphs.

1)\textbf{Placement graph} which is a bipartite graph whose a series of vertices are from the set of files $F=\{f_1,f_2,\cdots,f_{|F|}\}$, and the other ones are from the set of SBSs  $S=\{S_1,S_2,\cdots,S_{|S|}\}$. For $0\leq i\leq |F|$ and $0\leq j\leq |H|$, there is an edge between two vertices $f_i$ and $S_j$ if the file $f_i$ has been embedded in the  memory of SBS $S_j$. Fig. 2 is an example of a placement graph, with the assumption that all files can be embedded in all SBSs. However it is not a practical scenario, because the memory of the SBSs are limited, so a limited number of files can be stored in the memory of SBSs and therefore some edges are deleted.

In this graph, if the memory of  SBSs are equal, the corresponding degree of nodes will be equal. The vertices of the nodes in the file side are not equal, since the number of time for each file's placement in the SBSs are generally different.

2)\textbf{Access graph} which is a bipartite graph whose a series of vertices are from the set of users $U=\{u_1,u_2,\cdots,u_{|U|}\}$, and the other series of the vertices are from the set of SBSs $S=\{S_1,S_2,\cdots,S_{|S|}\}$. For $0\leq i\leq |U|$ and $0\leq j\leq |S|$, there is an edge between two vertices $u_i$ and $S_j$ if the user $u_i$ has access to the SBS's memory $S_j$. Fig. 2(b) is an example of an access graph, with the assumption that all users have access to all SBSs, which is not the case in practice.

3)\textbf{SBSs graph} which is a weighted graph in which SBSs are the vertices of that graph. This graph is not necessarily bipartite. The way to build this graph is to have one vertex for each SBS. There is an edge between every two vertices $S_i$ and $S_j$ where $0\leq i,j\leq |S|$, and we call it the weights of the edges, and show it with $w_{i,j}$.  As there is an edge between each vertex, this graph is a complete graph. An example of this graph is seen in Fig. 2(e).

4)\textbf{Delivery graph} which is a bipartite graph whose one series of vertices are from the set of files $F=\{f_1,f_2,\cdots,f_{|F|}\}$, and the other vertices are from the set of users $U=\{u_1,u_2,\cdots,u_{|U|}\}$. For $0\leq i\leq |F|$ and $0\leq j\leq |H|$, there is an edge between two vertices $f_i$ and $u_j$ if there is a path between two vertices $u_i$ and $f_j$ in the general graph of the network. An example of this graph is seen in Fig. 2(d).

The location of SBSs are fixed and users are moving.  We focus on the SBSs graph  in order to obtain the best placement graph, and then obtain the delivery graph according to the equivalent placement and access graphs. We should emphasize that our strategy is such that to be independent of the access graph since generally users are moving and so their corresponding graph is constantly changing. More precisely, the access graph is not the base of our caching strategy.

A matrix corresponding to each graph is defined as the adjacency matrix of the graph. Adjacency matrix of a simple graph is a square matrix of order $n\times n$ that entry of i-th raw and j-th column of this matrix is  1 if and only if there is an edge between the  i-th vertex and j-th vertex in the equivalent graph and otherwise is 0. Also, Adjacency matrix of a simple weighted graph is symmetric square matrix of order $n\times n$ that entry i-raw and j-column of this matrix is the weight  of edge between the  i-th vertex and j-th vertex in the equivalent graph.

\subsection{Graph coloring concepts}
Labeling a graph is the attribution of labels to the graph edges,  vertices or both of them, which it is typically done by assigning integers to them\cite{mj15}. The general approach is to use colors to match the edges or vertices that two adjacent edges do not have the same color. In the simplest case, the vertex coloring of a simple  un-weighted graph G(V,E) can be defined as
\begin{eqnarray}\label{vertec}
\begin{array}{cc}
f:V(G)\longrightarrow N\\
if ~~uv\in E(G)~~ then~~ f(u)\neq f(v).
\end{array}
\end{eqnarray}
To color a weighted graph, we can convert it to a simple un-weighted graph, then color it. Although the  edge coloring of a graph is similarly defined, but in the remainder of this paper where the color is mentioned, the color of the vertices is our intention.
 Coloring of a graph with a maximum of k colors is called k-coloring.  The smallest number of colors needed to color a graph G is called its chromatic number, and is often denoted as $X(G)$. A graph which can be colored with k colors is called k-colorable graph and a set of vertices of a graph that are colored with the same color forming a color class. Each color class forms an independent set. Colored polynomials refer to the number of coloring methods of a graph, in which the number of colors used does not exceed the specified number.

 \section{Caching through Graph coloring}\label{method1}
 In this section, we aim to apply graph coloring to cache the popular files in the memories exist in the network. First, we create a graph for our system and then we run the coloring on it. According to the graph model explained in section \ref{mostafa1000}, in order to obtain the appropriate graph, we need to make changes to the SBSs graph. We make these changes in order to obtain a proper placement graph by coloring the resulted graph.
 \subsection{Modification  of SBSs graph}\label{sepas1}
We consider equivalent adjacent matrix of a weighted graph on SBSs which for $1\leq i,j\leq |S|$, entry of  i-row and j-column is $w_{ij}\neq 0$.   This matrix is adjacency matrix of a weighted graph whose entries may not be zero and one. In order to convert this matrix to the adjacency matrix of an  un-weighted graph, we put threshold (Tr) on the entry of that matrix then we define the new weights as follows
\begin{eqnarray*}
w^\prime_{ij}=\left\{
\begin{array}{cc}
1, & w_{i,j}\geq Tr(i,j) \\
0 & w_{i,j}<Tr(i,j).
\end{array}
\right.
\end{eqnarray*}
Now draw the corresponding un-weighted new graph $G^\prime$ based of the new matrix $w^\prime$. In this matrix we have an edge between every two SBSs nodes with the weight larger than threshold level and  vice versa.
With this definition, we convert the weighted graph into a weightless graph because we need to have a weight-free graph to apply coloring on it. Also  applying a thresholds on the edges, we have more flexibility for coloring. Here, we consider the weights in the graph of SBSs equal to the distance between SBSs nodes(for both SBS nodes $S_i$ and $S_j$ where $0\leq i,j\leq |S|$ i. e. $w_{i,j}$=$d_{i,j}$). The threshold on the weights can be defined in both individual and universal cases. In individual mode, for each edge, we define the threshold value separately, but in the universal mode a threshold is defined for all the edges.  In this model, assuming that the range of the SBS $S_i$ for $1\leq i\leq |S|$ is equal to $R_i$ , we define the universal and individual thresholds as following
 \begin{eqnarray*}
Tr_(i,j)^{(IND)}=min\{R_i,R_j\}
\end{eqnarray*}
 \begin{eqnarray*}
Tr_(i,j)^{(UNI)}=min\{Tr_(i,j)^{(IND)}, for ~0\leq i,j\leq |S|\},
\end{eqnarray*}
where $Tr_(i,j)^{(IND)}$ and $Tr_(i,j)^{(UNI)}$ are individual and universal threshold between $S_i$ and and $S_j$ for $~0\leq i,j\leq |S|$, respectively.
 \subsection{Placement by coloring the graph}\label{methcol1}
 Depending on the system model, a user may simultaneously access several SBSs. In the case that  users only have access to a SBS, the best way to cache popular files  is to put most popular files into the cache of SBSs, but if a user simultaneously has access to two or more SBSs, placing non-common files on the SBSs results in the increment of the probability of access to more files.The graph-coloring algorithms are based on the following two conditions: 1)No two adjacent vertices have the same color 2)An algorithm that uses fewer colors is considerable (in other words, whatever the number of colors used in the algorithm is closer to $X(G)$, it is more desirable).  By applying the proper coloring algorithm, a user do not have access to two memories of SBSs  with common files and also access files with higher popularity.  In addition to these two conditions, we add the condition that  in the coloring, the priority is  with the vertex with higher degree. This additional condition helps us place popular files in the SBSs exist in the dense area, first.

In the following we present our proposed algorithm to make the placement graph. We consider the graph $G^\prime$  as the basis of algorithm \ref{algo1} and execute the algorithm steps on it.

In this algorithm, we receive the graph of SBSs constructed in Section \ref{sepas1} and  set of popular file $F=\{f_1,f_2,\cdots , f_{|F|}\}$ where $i\leq j$ then $p(f_i)\leq p(f_j)$, and  assign some of these files to SBS in order to build a placement graph. We apply the coloring algorithm, which has a color priority, to the color of the vertices of the larger degree on the graph. By putting this priority, we plan to first fill the memory of SBS that has more probability to be placed in a more crowded position.  Then we use vertex coloring, which ensures that two adjacent vertices do not have the same color. The condition of not having the same color for two adjacent vertices helps us to not place a similar file in SBSs that can serve a common user. The result of this, is a graph that assign a number  to each vertex where these numbers represent the priority to fill the memory of SBSs. Based on this placement matrix, the access matrix can be provided with the assumption that each user  in the range of each SBS can have access to its memory.
\begin{algorithm}\label{algo1}
\caption{ Placement based on Coloring}
\begin{algorithmic}
\floatname{algorithm}{Procedure}
\renewcommand{\algorithmicrequire}{\textbf{Input:}}
\renewcommand{\algorithmicensure}{\textbf{Output:}}
\REQUIRE Graph $G^\prime$ with vertex set $V(G^\prime)=\{v_1,v_2,\cdots , v_|S|\}$ and set of popular file $F=\{f_1,f_2,\cdots , f_{|F|}\}$  where $i\leq j$ then $p(f_i)\leq p(f_j)$\\
\FORALL {$v \in V(G^\prime)$}
\STATE assign proper color to v by vertex coloring algorithm based on the high degree of priority with index color $\{1,2,\cdots , X(G^\prime)\}$
\ENDFOR
\FORALL {$SBS \in S$}
\STATE Fill in the memory of SBSs based on the high priority of the corresponding color index
\ENDFOR
\ENSURE Placement Graph
\end{algorithmic}
\end{algorithm}

With respect to the SBSs graph of different networks, different $X(G)$ are obtained. Several upper and lower boundaries on $X(G)$ can be considered as follows.
  \begin{theorem}
 $\omega(G)\leq X(G)<\Delta(G)+1$, where $\Delta(G)$ is maximum degree of graph G and $\omega(G)$ is the largest set of mutually adjacent vertex in G.
 \end{theorem}
 \begin{proof}
 See in \cite{mj601,mj602,mj603}.\\
\end{proof}
\begin{theorem}
For every complete, interval  and bipartite graph $G$, $X(G)=\omega(G)$, where a graph is an interval graph if it has an intersection model consisting of intervals on a straight line.
 \end{theorem}
 \begin{proof}
 See in \cite{mj601,mj602,mj603}.\\
\end{proof}
\begin{theorem}
 For every graph G, $X(G)\geq \frac{n(G)}{\alpha(G)}$, where $\alpha(G)$ is the maximum size of vertices that are mutually non-adjacent and $n(G)$ is the number of vertices.
 \end{theorem}
 \begin{proof}
 See in \cite{mj601,mj602,mj603}.\\
\end{proof}

In our system model N(G) is the number of SBSs. The greater the density of the SBSs the larger $\Delta(G)$ and $\omega(G)$, and smaller  $\alpha(G)$ in the graph constructed in \ref{sepas1}. So the number of colors should be used to color, increases in our algorithm.

\subsection{Graph Coloring Problem and Formulating an Integer Program}
In the previous section, an algorithm was introduced which in one of its step, we colored graph with  a proper coloring method. A proper coloring is an optimization problem that is to find a vertex coloring with minimum number of colors. This optimization problem can be formulated as an integer programming\cite{mj15}. Suppose that $i\in V, 1\leq j\leq n$ where $n=|V|$. We consider the variable $x_{ij}$ as follow

 \[
 x_{ij}=
 \begin{cases}
 1 \quad  if~ color ~j ~is ~assigned~ to~ vertex~ i\\
 0\quad otherwise
 \end{cases}
 \]
 In the worst case, a graph can be colored by assigning different colors to the top of the graphs, so this variable is sufficient. The variable $w_j$ for $1\leq j \leq n$ is also defined as follow
  \[
 w_{j}=
 \begin{cases}
 1 \quad  if~ color ~j ~is ~used~ in~ coloring~ i\\
 0\quad otherwise
 \end{cases}
 \]
 In the other words, $w_{j}=1$ if there is a vertex  i such that $x_{ij}=1$. The integer programming model is expressed as follow  based on the mentioned variables.

 \begin{align*}
\hspace{-4.15cm}\mathop {\mathbf{min}}\,\,\,\,\sum_{j=1}^{n}w_j
\end{align*}

\begin{align*}
\textbf{\text{s.t.}}\hspace{1cm}&\mathbf{C\cdot4\cdot1}\hspace{.5cm}x_{ij}+x_{kj}\leq w_j, \forall ik\in E(1\leq j\leq n)\\
&\mathbf{C\cdot4\cdot2}\hspace{.5cm}x_{ij}\in \{0,1\}, \forall i\in V(1\leq j\leq n)\\
&\mathbf{C\cdot4\cdot3}\hspace{.5cm}w_{kj}\in \{0,1\}, (1\leq j\leq n). \\
&\hspace{5cm}\mathbf{(P\cdot4\cdot1)}
\end{align*}

 \begin{theorem}
 graph-coloring problem is NP-hard problem.
 \end{theorem}
 \begin{proof}
 See in \cite{mj16,mj17}.\\
\end{proof}
 This NP-hard problem can be solved in polynomials time for special classes of graphs, for example, perfect graphs \cite{mj15}.The simplest case for this problem is the coloring of complete graph, where the complete graph is a graph with at least one edge between every two vertices. For the complete graph, we need the largest color to color a graph, which is equal to the number of vertices of the graph. In other words, in the complete graph, we assigned all vertices with different colors, since in this graph, both vertices are adjacent.

Accurate and approximate algorithms are considered as two major approaches to solve graph coloring problems\cite{mj20}. Accurate algorithms can solve the coloring of graphs with a maximum of 100 vertices. In the case of coloring graphs with more than 100 vertices, approximate algorithms are needed. The Greedy constructive, Local Search heuristics and Meta heuristics are the examples of the approximate algorithm\cite{mj18,mj19}. DSATUR and RLF are the most important ways to  greedy algorithms that their outputs has recently been used as the initial solution for advanced Meta heuristic algorithms\cite{mj20}. Local Search algorithms  are
 used in hybrid methods that for example, the TS algorithm is one of them\cite{mj20}. The Clique Algorithm by Ashay Dharwadker is  an accurate algorithm  for vertex coloring\cite{mj14}.

\section{Caching 's strategy based on  clustering, categorizing and weighing SBS nodes}\label{method2}
In the previous section, caching was performed applying graph coloring on the graph of SBSs.  We constructed this graph using the threshold setting on a weighted graph where  the distance between the SBSs  used as the weights of the edges. This matrix can also be constructed in other ways. In this section, we intend to construct this matrix using two point processes so that each vertex is weighted.

\subsection{matern-type I and II process}\label{jk2}
The hard core process is a point process that results from a poisson process in which the parent points with a certain property attract the child's points. This process determines the center of non-overlapping objects, such as a circle  with radius $R_c$. These processes can be marked by marks that are independent of the process points and are i.i.d. as a variable.
An example of the hard core process is the matern hard core process that was introduced in 1960 and 1986 by Matern\cite{mj21}. The matern-type I and II processes are two kinds of matern hard core processes that we use them for the classification. If a circle with radius D has been considered around a point in a Poisson process  $N_b$ with density $\lambda_b$,  then the points in the circles do not overlap each other form the Matern Core-type I process. We present the resulting process with $N_b^{I}$ and define it as follow
\begin{equation}\label{jk1}
N_b^{I}=\{x\in N_b:B(x,D)\cap B(y,D)=\emptyset , ~\forall y\in N_b\}.
\end{equation}
 To form Matern Core-type II from a Poisson process $N_b$ with density $\lambda_b$, we are assigned randomized marks distributed uniformly between zero and one to Poisson process points. The points where their marks are larger than their neighbors in distance D from are removed. The points that remain will form a Matern Core-type II process. We present the resulting process with $N_b^{II}$ and define it as follow
\begin{equation}\label{jk1}
N_b^{II}=\{x\in N_b:m(x)<m(y),~\forall y\in N_b\cap B(x,D)\setminus\{x\}\}.
\end{equation}

We use the matern-type II process to cluster SBS nodes that cover the same users. The procedure is done in such a way that  the location of the SBS nodes is a Poisson point process with density $\lambda_b$. Using this process and assuming $ِD=2R_c$ we get matern-type II process with point in $N_b^{II}$. Consider the obtained process points as the center of the clusters in the radius $D=2R_c$. This kind of clustering is called clustering type II.
\subsection{Classifying and weighing SBSs}
We want to classify SBSs and assign a weight to each one so that SBSs that fall into a class have a special feature, and SBSs with different weights have different values for us. To achieve this,
we apply the Matern Core-type I to clustering the SBSs who do not have any SBS in their coverage range of them. In other words, when a user is in the range of this SBS, it can only be serviced by that SBS and has no choice to take services from the other SBSs. It is therefore reasonable that such a SBS be filled with the most popular files. The Matern Core-type II clusters SBSs which are in the dense area. In such an environment, the user has access to more than one SBS simultaneously. Therefore, in these scenarios, we try to put more files in SBSs that each user access them at any time by avoiding placing duplicate files in the memories of neighbor SBSs. In addition to use two processes to determine the dense area, we use them to weight SBSs. This weightening is due to the fact that SBSs in more dense areas with larger weights get more popular files because they are placed in places where the user density may be higher.
This procedure goes as algorithm \ref{jk4}.
\begin{algorithm}\label{jk4}
\caption{classifying and weighing  of SBSs}
\begin{algorithmic}
\floatname{algorithm}{Procedure}
\renewcommand{\algorithmicrequire}{\textbf{Input:}}
\renewcommand{\algorithmicensure}{\textbf{Output:}}
\REQUIRE SBS Position as a Point Poisson Process $N_b$ and range of classification($R_c$)\\
\FORALL {$S_i \in S, 1\leq i\leq |S|$ }
\STATE $W_{S_i}\gets 0$
\STATE $D_{S_i}\gets \{S_j; d(S_i,S_j)\leq R_c, 1\leq j\leq |S|\}$
\ENDFOR

\WHILE {$W_{S_i}\neq \emptyset, 1\leq i\leq |S|$}
\STATE apply Matern Core type-I and II and achieve $N_b^{I}$ and $N_b^{II}$
\FORALL {$S_i$ in location $y \in N_b^{I}\cup N_b^{I}$}
\FORALL {$S_j\in D_{S_i}$}
\STATE $W_{S_j}\gets W_{S_j}+1$
\ENDFOR
\ENDFOR
\ENDWHILE
\ENSURE $ D_{S_i}$ as SBSS co-class with $ {S_i}$, for $1\leq i\leq |S|$ and $W_{S_i}$ as weigh for $ {S_i}$.
\end{algorithmic}
\end{algorithm}
In this algorithm, we receive the SBSs positions as a point poisson process $N_b$ and range of classification($R_c$). The distance between SBSs is used as a classification parameter.  First, the distance between each SBS and the other SBSs are calculated, and SBSs that are within the distance of less than the range of classification are selected as the co-class of this SBS. The initial weight of each SBS  set is zero. This method is applied on SBSs until the weights of all SBSs is non zero. In this method, the points of the processes of Matern Core type-I and II are selected, and the weight of the SBSs placed in the  SBS class corresponds to the selected points increases one unit. At the end, the outputs are $ D_{S_i}$ as SBSs who are co-class with $ {S_i}$, for $1\leq i\leq |S|$ and $W_{S_i}$ as the weight for $ {S_i}$.
\subsection{Coloring of the constructed graph of SBS}
Now when SBSs are classified and weighted, we intend to construct graph of SBSs using this classification and weightening so that we can apply coloring  on it. We assume that $\{C_1,C_2,\cdots ,C_{l}\}$ are the classes made by the algorithm 2.  We make  graph $G=(V(G),E(G,I_G)$. The set $V(G)=\{v_1,v_2,\cdots,v_{|S|\}}$ are the vertices of this graph that for each SBS, we have an equivalent vertex of the same index.  E(G) are the edges of this graph.  $I_G$ is a map that matches each member of V(G) to a non-order pair of vertices in V(G).

\begin{eqnarray}\label{jk6}
 \left\{
\begin{array}{ll}
I_G:E(G)\longrightarrow V(G)\\
I_G(e_k)=(v_i,v_j)~~if ~~v_i, v_j~~ in ~~same ~~class
\end{array}
\right..
\end{eqnarray}
In the other word, we have an edge between every two vertices if both vertices are placed in the same class. The vertices of this graph also have the corresponding weights derived from the algorithm 2. The placement graph  is constructed using algorithm 3  based on the obtained graph  in (\ref{jk6}).
\begin{algorithm}\label{jk5}
\caption{Coloring and Placement based on Coloring}
\begin{algorithmic}
\floatname{algorithm}{Procedure}
\renewcommand{\algorithmicrequire}{\textbf{Input:}}
\renewcommand{\algorithmicensure}{\textbf{Output:}}
\REQUIRE constructed graph  in \ref{jk6}\\
\WHILE {Every vertex has color}
\FORALL {$v_i \in V(G), 1\leq i\leq |S|$ }
\STATE By priority from the vertices with the highest weights, assign each vertex the smallest color that has not a neighbor vertex with that color.
\ENDFOR
\ENDWHILE
\FORALL {$SBS \in S$}
\STATE Fill in memory of SBS based on the high priority of the corresponding color index
\ENDFOR
\ENSURE Placement Graph
\end{algorithmic}
\end{algorithm}
 Vertex coloring is done on this graph with three conditions. 1)the coloring priority is with a vertex that has a higher weight. 2)to color each vertex, the coloring is applied in such a way  that the neighbors do not have the same color 3)to color each vertex, the smallest possible color is used. Our goal is to place the most popular files in the memory of SBSs so that the SBSs who have overlap coverage range, cache the separate files and at the same time the most popular files are embedded in SBSs leads to the serving most of the requested files. Based on the obtained placement graph and the access graph, which is changing due to the movement of the users, the delivery graph is made. This graph specifies the files available to each user at the  delivery phase according to its position.
\section{Simulation Results}
In this section, we intend to evaluate the performance of the proposed methods. Consider a circular cell with a radius of 350 meters \cite{mj10}. Suppose a single MBS covers all of this cell and $|S|$ SBSs is randomly distributed  in this cell so that each SBS has a  cover range of 80 meters. A total of 1,000 mobile users are randomly located in this cell and they are mobile. The system supports a collection of 1000 unit-sized files that have a popularity distribution of Zipf with parameter $\alpha$ \cite{mj10}. Each SBS can store 50 files, which is equivalent to 0.05 of total files.

To illustrate the effectiveness of the proposed methods, we calculate the hit rate value, which is the ratio of the number of requests answered  with cache of SBS to the total requests. We have calculated performance of the proposed methods in terms of the number of SBSs and Zipf distribution parameter ($\alpha$).

The performance of the proposed methods are presented in Fig. 3. We assume that the parameter of Zipf distribution ($\alpha$) is 0.6 and derived the hit rate versus  the number of SBSs in the network.   The curves in Fig. 3. are for three cases

 1)The memory of each SBS is cached which the most popular files, which has the role of the benchmark in our evaluation and is  as the baseline popularity in figure

    2)The memory of each SBS is cached with the method presented in Section \ref{method1}, which is described as the proposed content placement 1 in the figure

      3)The memory of each SBS is cached with the method presented in Section \ref{method2} , which is described as the proposed content placement 2 in the figure.

      As shown in Fig. 3., the two proposed methods have better performance compared to the base method and  improved the network performance from the hit rate aspect via simulation.

\begin{figure}[t!]\label{simBs}
\begin{center}
\includegraphics[width=0.53\textwidth]{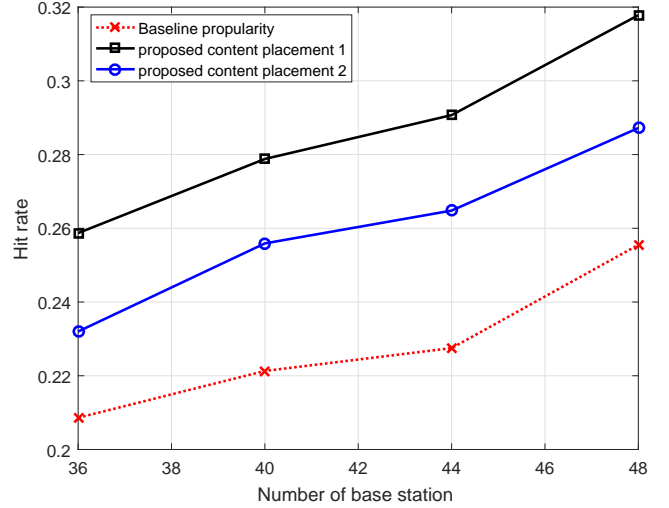}
\caption{Hit rate versus the number of SBSs.}
\end{center}
\end{figure}

In Fig. 4. the performance of the network has been investigated in different popularity distribution of files. We assumed that the number of SBSs in the network are 48.  The results are evaluated as same as the three cases mentioned in previous figure.  As this figure shows, the two proposed methods have better performance compared to the benchmark (base method).
\begin{figure}[t!]\label{SimZip}
\begin{center}
\includegraphics[width=0.53\textwidth]{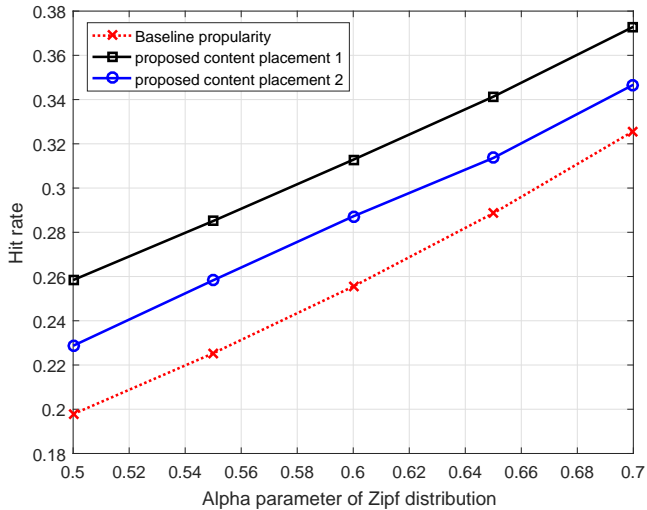}
\caption{Hit rate versus the Zipf distribution parameter ($\alpha$).}
\end{center}
\end{figure}

It should be mentioned that, although the results of the second method are not as good as the first one, its advantage is its simplicity in the coloring process compared to  the first method. In the first method, the graph is colored by solving an integer programming problem which is an NP-hard problem, but in the second method, the planning problem is replaced by a simple coloring method based on the weight assigned to the SBSs. This simple coloring method may uses different colors than the solution obtained by solving the optimizing integer programming problem of  graph coloring, but the simplicity of the method is its wealth.

From a different viewpoint, in order to compare the performance of the two described thresholding methods, consider the previous network and assume that SBSs have a random coverage range between 50 and 100 meters. The results obtained in both  individual and universal thresholding are given in the Fig. 5 compared with  the random most populared caching case.
\begin{figure}[t!]\label{thresh}
\begin{center}
\includegraphics[width=0.53\textwidth]{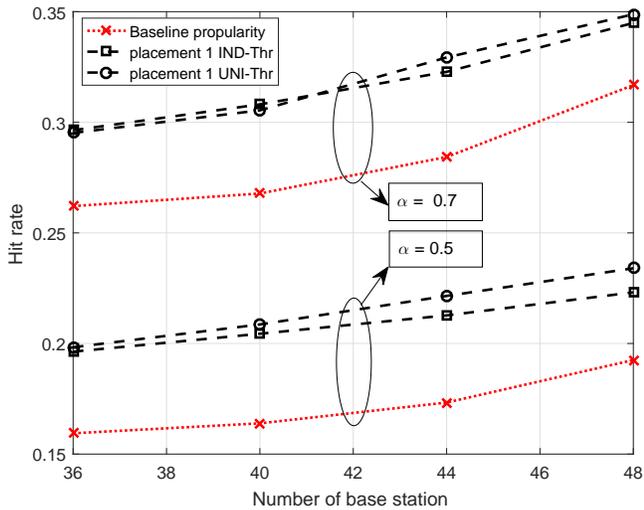}
\caption{Hit rate versus the number of SBSs for different Zipf parameter ($\alpha$).}
\end{center}
\end{figure}
\section{conclusion}
In this paper, we have considered a SBS network and studied the reduction of  the traffic load on the MBS using  our proposed caching strategy. We have mapped  the problem of  caching popular files to a graph which is divided into four new sub-graphs serving as the foundation of graph coloring in our two  presented methods.  In the first method, we  considered caching strategy  as a  NP-hard coloring problem and an  algorithm is introduced to color the graph or equivalently to cache the files. In the second method, using two point processes Matern Core-type I and II, all SBSs have been classified and weighted. The classification formed the SBSs graph and the weights which determined the priority of coloring in that graph.


%

\ifCLASSOPTIONcaptionsoff
  \newpage
\fi



%


\begin{thebibliography}{1}
\bibitem{mj4}
www.cisco.com/c/en/us/solutions/collateral/service-provider/visual-networking-index-vni/mobile-white-paper-c11-520862.html/

\bibitem{mj3}
N. Golrezaei, A. F. Molisch, A. G. Dimakis, and G. Caire,
``Femtocaching and device-to-device collaboration: A new architecture for wireless video distribution",
{\em  IEEE Commun. Mag.,} vol. 51, no. 4, pp. 142-149, Apr. 2013.

\bibitem{mj301}
N. Golrezaei, A. G. Dimakis, and A. F. Molisch,
``Scaling behavior for device-to-device communications with distributed caching",
{\em IEEE Trans. Inform. Theory,} vol. 60, no. 7, pp. 4286-4298, Jul. 2014.
\bibitem{mj302}
 N. Golrezaei, A. F. Molisch, and A. G. Dimakis,
``Base-station assisted device-to-device communications for high-throughput wireless video networks",
 {\em  IEEE Trans. Wireless Commun.,} vol. 13, no. 7, pp. 3665-3676, Jul. 2014.



\bibitem{mj401}
M. Gregori, J. Gómez-Vilardebó, J. Matamoros, D. Gündüz,
``Wireless Content Caching for Small Cell and D2D Networks",
 {\em  IEEE Journal on Selected Areas in Commun.,} vol. 34, no. 5, pp. 1222-1234, May 2016.

\bibitem{mj1}
E. Bastug, M. Bennis, M. Kountouris, and M. Debbah,
``Cache-enabled small cell networks: modeling and tradeoffs",
 {\em  EURASIP Journal on Wireless Commun. and Networking,} vol. 2015, no. 1, pp. 1-11, Feb. 2015.
\bibitem{kd2}
K. Shanmugam, N. Golrezaei, A. G. Dimakis, A. F. Molisch, and G. Caire,
``Femtocaching: Wireless content delivery through distributed caching helpers",
 {\em  IEEE Trans. on Inform. Theory,} vol. 59, no. 12, pp. 8402-8413, Dec. 2013.

\bibitem{mj6}
M. A. Maddah-Ali and U. Niesen,
``Fundamental limits of caching",
 {\em  , IEEE Trans. on Inform. Theory}, vol. 60, no. 5, pp. 2856-2867, 2014.

\bibitem{mj501}
Mohammad Ali Maddah-Ali and Urs Niesen,
``Decentralized Coded Caching Attains Order-Optimal Memory-Rate Tradeoff",
 {\em  IEEE/ACM Trans. on Networking,} vol. 23,  no. 4, Aug. 2015.

\bibitem{mj1400}
https://en.wikipedia.org/wiki/Four$\_$color$\_$theorem


\bibitem{mj20}
S. Mahmoudi, S. Lotfi,
``Modified cuckoo optimization algorithm (MCOA) to solve graphcoloring problem",
 {\em  Applied Soft Computing,} vol. 33, Aug. 2015, PP 48-64.

\bibitem{mj14}
http://www.dharwadker.org/vertex$\_$coloring

\bibitem{kd1}
Z. Chen, N. Pappas,  M. Kountouris,
``Probabilistic Caching in Wireless D2D Networks: Cache Hit Optimal Versus Throughput Optimal",
 {\em IEEE  Commun. Lett.,} vol. 21, no. 3, Mar. 2017.
\bibitem{kd3}
J. Li, Y. Chen, Z. Lin, W. Chen, B. Vucetic, and L. Hanzo,
``Distributed caching for data dissemination in the downlink of heterogeneous networks",
 {\em IEEE Trans. Commun.,}  vol. 63, no. 10, pp. 3553-3568, Oct. 2015.

\bibitem{kd4}
Y. Guo, L. Duan, and R. Zhang.,
``Cooperative Local Caching Under Heterogeneous File Preferences",
 {\em IEEE Trans. Commun.,}  vol. 65, no. 1, pp. 444-457, Jan. 2017.

\bibitem{kd5}
J. Song, H. Song, and W. Choi,
``{Optimal caching placement of caching system with helpers}",
 {\em  in Proc. IEEE Int. Conf. Commun. (ICC),}  London, U. K., Jun. 2015, pp. 1825-1830.

\bibitem{kd6}
H. J. Kang and C. G. Kang,
``Mobile device-to-device (D2D) content delivery networking: A design and optimization framework",
 {\em J. Commun. Netw.,}  vol. 16, no. 5, pp. 568-577, Oct. 2014.

\bibitem{kd7}
J. Rao, H. Feng, C. Yang, Z. Chen, and B. Xia,
``Optimal caching placement for D2D assisted wireless caching networks",
 {\em in Proc. IEEE ICC,}  pp. 1-6, May 2016.

\bibitem{kd8}
D. Malak, M. Al-Shalash, and J. G. Andrews,
``Optimizing Content Caching to Maximize the Density of Successful Receptions in Device-to-Device Networking",
 {\em IEEE Trans. Commun.,} vol. 64, no. 10, pp. 4365-4380, Oct. 2016.




\bibitem{kd9}
S. H. Chae and W. Choi,
``Caching placement in stochastic wireless caching helper networks: Channel selection diversity via caching",
 {\em IEEE Trans. Wireless Commun.,}  vol. 15, no. 10, pp. 6626-6637, Oct. 2016.

\bibitem{kd10}
F. Gabry, V. Bioglio and I. Land,
``On Energy-Efficient Edge Caching in Heterogeneous Networks",
 {\em IEEE Journal on Selected Areas in Commun.,}  vol. 34, no. 12, pp. 3288-3298, Dec. 2016.

\bibitem{kd11}
K. Poularakis and L. Tassiulas,
``On Energy-Efficient Edge Caching in Heterogeneous Networks",
 {\em IEEE Trans. Wireless Commun.,}  vol. 64, no. 5, pp. 2092-2103, May. 2016.

 \bibitem{kd12}
B. Liu, H. Zhang, H. Ji, and X. Li,
``A novel joint transmission and caching optimizing scheme in multirelay networks: Video service quality assurance scheme",
 {\em Int. J. Commun. Syst.,}  2017, p. e3284. [Online]. Available: https://doi.org/10.1002/dac.3284.

\bibitem{kd13}
S. H. Chae,  T. Q. S. Quek  and  W. Choi,
``Content Placement for Wireless Cooperative Caching Helpers: A Tradeoff Between Cooperative Gain and Content Diversity Gain",
 {\em IEEE Trans. Wireless Commun.,}  vol. 16, no. 10, pp. 6795-6807, Oct. 2017.

\bibitem{kd14}
A. Liu and V. K. N. Lau,
``Cache-enabled opportunistic cooperative
MIMO for video streaming in wireless systems",
 {\em IEEE Trans. Signal Process.,}  vol. 62, no. 2, pp. 390-402, Jan. 2014.

\bibitem{kd15}
H. Li, C. Yang, X. Huang, N. Ansari, and Z. Wang,
``Cooperative RAN Caching Based on Local Altruistic Game for Single and Joint Transmissions",
 {\em IEEE  Commun. Lett.,}  vol. 21, no. 4, pp. 853-856, Apr. 2017.

\bibitem{kd16}
J. Liu, B. Bai, J. Zhang, , and K. B. Letaief,
``Cache Placement in Fog-RANs: From
Centralized to Distributed Algorithms",
 {\em IEEE Trans. Wireless Commun.,}  vol. 16, no. 11, pp. 7039-7051, Nov. 2017.

\bibitem{kd17}
H. Ahlehagh and S. Dey,
``Video-aware scheduling and caching in the radio access network",
 {\em IEEE/ACM Trans. Netw.,}  vol. 22, no. 5, pp. 1444-1462, Oct. 2014.

\bibitem{kd18}
J. Song,  H. Song, and W. Choi,
``Optimal Content Placement for Wireless
Femto-Caching Network",
 {\em IEEE Trans. Wireless Commun.,}  vol. 16, no. 7, pp. 1536-1276, Jul. 2017.

\bibitem{kd19}
B. Hong and W. Choi,
``Optimal storage allocation for wireless cloud caching systems with a limited sum storage capacity",
 {\em IEEE Trans. Wireless Commun.,}  vol. 15, no. 9, pp. 6010-6021, Sept. 2016.

\bibitem{kd20}
R. Wang, X. Peng, J. Zhang, and K. B. Letaief,
``Mobility-aware caching for content-centric wireless networks: Modeling and methodology",
 {\em IEEE Commun. Mag.,}  vol. 54, no. 8, pp. 77-83,
Aug. 2016.



\bibitem{kd21}
I. Keshavarzian,  Z. Zeinalpour-Yazdi,
and A. Tadaion,
``A clustered caching placement in heterogeneous small cell networks with user mobility",
 {\em in Proc.  IEEE International Symposium on Signal Processing and Information Technology (ISSPIT),}    2015, pp. 421 - 426.

\bibitem{kd22}
J. Wen, K. Huang,  S. Yang, and V. O. K. Li,
``Cache-Enabled Heterogeneous Cellular Networks:
Optimal Tier-Level Content Placement",
 {\em IEEE Trans. Wireless Commun.,}  vol. 16, no. 9, pp. 5939-5952, Sept. 2017.


\bibitem{kd23}
T. Yang, R. Zhang, X. Cheng and L. Yang,
``Graph Coloring Based Resource Sharing (GCRS) Scheme for D2D Communications Underlaying Full-Duplex Cellular Networks",
 {\em IEEE Transactions on Vehicular Technology,}  vol. 66, no. 8, pp. 7506-7517, Aug. 2017.

\bibitem{kd24}
D. Tsolkas, E. Liotou, N. Passas, and L. Merakos,
``A graph-coloring secondary resource allocation for D2D communications in LTE networks",
 {\em in Proc. IEEE 17th Int. Workshop Comput. Aided Model. Des. Commun. Links Netw. Conf.,}  Barcelona, Spain, Sep. 17-19, 2012, pp. 56–60.

\bibitem{kd25}
X. Cai, J. Zheng, and Y. Zhang,
``A graph-coloring based resource allocation
algorithm for D2D communication in cellular networks",
 {\em in Proc. IEEE Int. Conf. Commun.,}  London, U.K., Jun. 8-12, 2015, pp. 5429–5434.




\bibitem{mj15}
P. Coll, J. Marenco, I. M. Díaz, P. Zabala,
``Facets of the graph coloring polytope",
 {\em  Annals of Operations Research,} 116 (1-4) (2002) 79-90.



\bibitem{mj601}
D. B. West,
``Fundamental limits of caching",
 {\em Prentice Hall,} 1996, 2001.
\bibitem{mj602}
J. M. Aldous and R. J. Wilson,
``Graphs and Applications:An Introductory Approach",
 {\em London: Springer,} 2000.

\bibitem{mj603}
J. A. Bondy and U. S. R. Murty
``Graph Theory With Applications",
 {\em Elseyier Science Publishing Co., Inc.,} 1976.



\bibitem{mj10}
K. Poularakis, G. Iosifidis and L. Tassiulas,
``Approximation algorithms
for mobile data caching in small cell networks",
 {\em   IEEE Transactions on Commun.,} vol. 62, no. 10, pp. 3665-3677, Oct. 2014.







\bibitem{mj16}
M. R. Carey and D. S. Johnson,
``The complexity of near-optimal graph coloring",
 {\em  J. Assoc. Comput. Mach.,} 23 (1976) 43-49.

\bibitem{mj17}
M. R. Garey, D. S. Johnson and L. Stockmeyer,
``Some simplified NP-complete graph problems",
 {\em  Theoret. Comput. Sci.,} 1 (1976) 237-267.



\bibitem{mj18}
D. d. Werra
``Heuristics for graph coloring, Computing,",
 {\em  Computing, Suppl.,} 7 (1990) 191-208.

\bibitem{mj19}
F. Glover, M. Parker and J. Ryan,
``Coloring by tabu branch and bound, in: Cliques",
 {\em  Coloring, and
Satisfiability, DIMACS Series on Discrete Mathematics and Theoretical Computer Science,} Vol. 26,
eds. D. Johnson and M. Trick (1996) pp. 285-308.



\bibitem{mj21}
S. N. Chiu, D. Stoyan, W. S. Kendall and J. Mecke,
``Stochastic geometry and its applications",
{\em  John Wiley \& Sons}, 2013.


\bibitem{mj21}
S. N. Chiu, D. Stoyan, W. S. Kendall, and J. Mecke, ``Stochastic geometry and its applications",{\em  John Wiley \& Sons}, 2013.


\end{thebibliography}

\end{document}